\documentclass[10pt]{amsart}
\usepackage[margin=3.5cm]{geometry}
\usepackage{amssymb}
\usepackage{amsmath}
\usepackage{centernot}
\usepackage{graphicx} 
\usepackage[mathscr]{euscript}
\usepackage{enumerate}
\usepackage{xspace}
\usepackage{color}
\usepackage{enumitem}
\usepackage{amsthm}
\usepackage{dsfont}
\usepackage{mathrsfs}
\usepackage[scr=boondoxo]{mathalfa}
\usepackage{bbm}
\usepackage[colorlinks=true, linkcolor = blue, citecolor = blue]{hyperref}
\usepackage[numbers]{natbib}
\usepackage[backgroundcolor=white,bordercolor=red]{todonotes}
\DeclareMathAlphabet{\mathpzc}{OT1}{pzc}{m}{it}
\usepackage[nameinlink]{cleveref}

\begin{document}

\theoremstyle{plain}
\newtheorem{theorem}{Theorem}[section]
\newtheorem{lemma}[theorem]{Lemma}
\newtheorem{proposition}[theorem]{Proposition}
\newtheorem{corollary}[theorem]{Corollary}
\newtheorem{Ass}[theorem]{Assumption}

\theoremstyle{definition}
\newtheorem{discussion}[theorem]{Discussion}
\newtheorem{definition}[theorem]{Definition}
\newtheorem{remark}[theorem]{Remark}
\newtheorem{example}[theorem]{Example}
\newtheorem{condition}[theorem]{Condition}
\newtheorem{agreement}[theorem]{Agreement}
\newtheorem{SA}[theorem]{Standing Assumption}
\newtheorem*{observation}{Observations}
\newtheorem*{RL}{Comments on Related Literature}
\newtheorem*{setting}{Setting}

\renewcommand{\chapterautorefname}{Chapter} 
\renewcommand{\sectionautorefname}{Section} 

\crefname{lemma}{lemma}{lemmas}
\Crefname{lemma}{Lemma}{Lemmata}
\crefname{corollary}{corollary}{corollaries}
\Crefname{corollary}{Corollary}{Corollaries}

\newcommand{\m}{\mathfrak{m}}
\newcommand{\mo}{\mathfrak{m}^\circ}
\newcommand{\mon}{\mathfrak{m}^{\circ, n}}
\newcommand{\tm}{\tilde{\mathfrak{m}}}
\newcommand{\M}{\tilde{\mathfrak{m}}}
\newcommand{\Mo}{\tilde{\mathfrak{m}}^\circ}
\newcommand{\Mon}{\tilde{\mathfrak{m}}^{\circ, n}}
\newcommand{\s}{\mathfrak{s}}
\newcommand{\ts}{\tilde{\mathfrak{s}}}
\renewcommand{\S}{S} 
\renewcommand{\v}{\mathfrak{v}}
\renewcommand{\u}{\mathfrak{u}}
\newcommand{\q}{\mathfrak{q}}
\def\stackrelboth#1#2#3{\mathrel{\mathop{#2}\limits^{#1}_{#3}}}
\renewcommand{\l}{\mathscr{l}}
\newcommand{\E}{\mathsf{E}} 
\renewcommand{\P}{\mathsf{P}}
\newcommand{\Po}{\mathds{P}^\circ}
\newcommand{\Pon}{\mathds{P}^{\circ, n}}
\newcommand{\Q}{\mathsf{Q}} 
\newcommand{\tP}{\tilde{\mathsf{P}}}
\newcommand{\tQ}{\tilde{\mathsf{Q}}}
\newcommand{\Qo}{\tilde{\mathds{P}}^\circ}
\newcommand{\Qon}{\tilde{\mathds{P}}^{\circ, n}}
\newcommand{\W}{\mathds{W}}
\newcommand{\on}{\operatorname}
\newcommand{\oU}{U}
\newcommand{\of}{[\hspace{-0.06cm}[}
\newcommand{\gs}{]\hspace{-0.06cm}]}
\newcommand{\ofr}{(\hspace{-0.09cm}(}
\newcommand{\gsr}{)\hspace{-0.09cm})}
\renewcommand{\emptyset}{\varnothing}

\renewcommand{\theequation}{\thesection.\arabic{equation}}
\numberwithin{equation}{section}

\newcommand{\1}{\mathds{1}}
\newcommand{\f}{\mathfrak{f}}
\newcommand{\g}{\mathfrak{g}}
\newcommand{\e}{\mathfrak{e}}
\renewcommand{\t}{T}
\newcommand{\bR}{\mathbb{R}}
\newcommand{\cF}{\mathcal{F}}
\newcommand{\cG}{\mathcal{G}}
\newcommand{\cA}{\mathcal{A}}
\newcommand{\x}{\mathfrak{z}}

\newcommand*{\Y}{Y}
\newcommand*{\U}{U}
\newcommand*{\Z}{Z}
\newcommand*{\B}{B}
\newcommand*{\V}{V}

\newcommand*{\ol}{\overline}
\newcommand*{\wt}{\widetilde}
\newcommand*{\canOmega}{\ol\Omega}
\newcommand*{\canF}{\overline{\mathcal{F}}\vphantom{()}}
\newcommand*{\canbfF}{\ol{\mathbf F} \vphantom{()}}
\newcommand*{\canY}{\ol\Y}
\newcommand*{\canU}{\ol\U}
\newcommand*{\canV}{\ol\V}
\newcommand*{\canZ}{\ol\Z}
\newcommand*{\cantheta}{\bar\theta}
\newcommand*{\wtY}{\wt\Y}
\newcommand*{\wtZ}{\wt\Z}
\newcommand{\ma}{\m_{\on{a}}}
\newcommand{\ms}{\m_{\on{s}}}
\newcommand{\X}{X}

\newcommand*{\siQ}{{}^\sigma\hspace{-1pt}\Q}
\newcommand*{\siB}{{}^\sigma\mathbb B}
\newcommand*{\siS}{{}^\sigma\!\S}
\newcommand*{\sibfF}{{}^\sigma\mathbf F}
\newcommand*{\sicF}{{}^\sigma\!\cF}

\title[NA and ACLMMs for diffusions]{No arbitrage and the existence of ACLMMs\\in general diffusion models}

\author[D. Criens]{David Criens}
\address{D. Criens - University of Freiburg, Ernst-Zermelo-Str. 1, 79104 Freiburg, Germany.}
\email{david.criens@stochastik.uni-freiburg.de}

\author[M. Urusov]{Mikhail Urusov}
\address{M. Urusov - University of Duisburg-Essen, Thea-Leymann-Str. 9, 45127 Essen, Germany.}
\email{mikhail.urusov@uni-due.de}


\makeatletter
\@namedef{subjclassname@2020}{\textup{2020} Mathematics Subject Classification}
\makeatother



\date{\today}

\allowdisplaybreaks

\begin{abstract}
	In a seminal paper, F.~ Delbaen and W.~Schachermayer proved that the classical NA (``no arbitrage'') condition implies the existence of an ``absolutely continuous local martingale measure'' (ACLMM). It is known that in general the existence of an ACLMM alone is not sufficient for NA. In this paper we investigate how close these notions are for single asset general diffusion market models. 
	We show that NA is equivalent to the existence of an ACLMM plus a mild regularity condition on the scale function and the absence of reflecting boundaries. For infinite time horizon scenarios, the regularity assumption and the requirement on the boundaries
can be dropped, showing equivalence between NA and the existence of an ACLMM. By means of counterexamples, we show that our characterization of NA for finite time horizons is sharp in the sense that neither the regularity condition on the scale function nor the absence of reflecting boundaries can be dropped.
\end{abstract}

\maketitle

\frenchspacing
\pagestyle{myheadings}


\section{Introduction}
Deciding about the existence of arbitrage opportunities is a fundamental problem in mathematical finance.
The concept that appears to be closest to the idea of ``making something out of nothing'' is the notion NA (``no arbitrage''). The celebrated fundamental theorem for asset pricing, which is due to J.~M.~Harrison and D.~M.~Kreps \cite{HK79} and J.~M.~Harrison and S.~R.~Pliska \cite{HP81} for finite probability spaces and due to R.~C.~Dalang, A.~Morton and W.~Willinger \cite{DMW90} for general discrete time models
(with a finite time horizon),
shows that NA is equivalent to the existence of an equivalent martingale measure (EMM).
In continuous time, the theory of arbitrage turns out to be significantly different from its discrete time counterpart (compare Chapters~V and~VII in the monograph \cite{shir} by A.~N.~Shiryaev).
In particular, the equivalence between NA and the existence of an EMM 
fails for general continuous time market models.
As a remedy, F.~Delbaen and W.~Schachermayer introduced the nowadays classical notion NFLVR (``no free lunch with vanishing risk'') for general continuous time semimartingale financial market models and proved it to be equivalent to the existence of an equivalent \(\sigma\)-martingale measure (see their monograph~\cite{DS2006} for an overview on what they achieved).
For general path-continuous semimartingale market models, F.~Delbaen and W.~Schachermayer~\cite{DS1995} also showed that NA implies the existence of an absolutely continuous local martingale measure (ACLMM). The converse is false in general. 
Still for path-continuous models, Y.~Kabanov and C.~Stricker~\cite{KabStr2005} and E.~Strasser~\cite{Strasser2005} filled this gap by proving that NA is in fact equivalent to the existence of ACLMMs
(with certain additional properties)
for all randomly shifted market models.
Although NA might not be equivalent to the existence of an ACLMM, it is interesting to understand how close these concepts are and to identify frameworks for which equivalence holds. 

\smallskip
In this paper, we consider so-called general diffusion models, which are single asset market models whose (discounted) asset price process is a one-dimensional path-continuous regular strong Markov process that is also a semimartingale. This framework is very rich. For example, it covers all It\^o diffusion models of the type 
\begin{align} \label{eq: intro SDE}
	d Y_t = \mu (Y_t) dt + \sigma (Y_t) dW_t,
\end{align}
under the famous Engelbert--Schmidt conditions,
but also models with local time effects such as skewness and stickiness.
Similar to SDEs of the type \eqref{eq: intro SDE}, general diffusion models can be characterized by two deterministic objects, the scale function and the speed measure (cf., e.g., \cite{breiman1968probability,freedman,RW2}), which we call the diffusion characteristics.
In the previous paper \cite{criensurusov23}, we proved deterministic characterizations of NA in terms of the diffusion characteristics. 

The purpose of the present paper is to relate NA to the non-deterministic criterion of the existence of an ACLMM. For this purpose, we distinguish between a finite and the infinite time horizon. For the former, we prove that NA is equivalent to the existence of an ACLMM plus the additional conditions that the scale function is a \emph{dc function} (i.e., the difference of two convex functions) and all accessible boundaries of the diffusion are absorbing. In case the time horizon is infinite, we even establish the equivalence of NA and the existence of an ACLMM. By means of counterexamples, we show that this is a feature of the global setting, i.e., for finite time horizons the additional conditions on the scale function and the boundaries cannot be dropped.

Our results show that there are numerous frameworks for which NA and the existence of an ACLMM are in fact equivalent. A particularly interesting class is given by It\^o diffusion models of the type~\eqref{eq: intro SDE} whose accessible boundaries are stipulated to be absorbing. Indeed, in such cases the scale function is even \(C^1\) with absolutely continuous derivative (and consequently, a dc function).

\smallskip
The remainder of this paper is organized as follows.
In Section~\ref{sec: NA cond}
we recall the NA notion,
the theorem of Delbaen--Schachermayer that NA implies the existence of an ACLMM,
the characterization of NA by Kabanov--Stricker and Strasser, and we provide an explicit counterexample showing that, in general, the existence of an ACLMM is strictly weaker than NA.
The general diffusion framework and our main results are presented in Section~\ref{sec: main} (canonical setting) and Section~\ref{sec: main2} (non-canonical setting).

\section{The NA condition} \label{sec: NA cond}
Let \(\mathbb{B} = (\Omega, \cF, (\cF_t)_{t \geq 0}, \P)\) be a probability space
with a right-continuous filtration
that supports a one-dimensional continuous semimartingale~\(\S\), which plays the role of a discounted asset price process.
We call the pair \((\mathbb{B}, \S)\) a {\em financial market}. 
Let \(L (\S)\) be the set of all predictable real-valued processes \(H = (H_t)_{t \geq 0}\) which are integrable
w.r.t. \(\S\).
In our financial context, the elements of \(L (\S)\) are called \emph{trading strategies}. To ease our presentation, we write 
\[
V^H \triangleq \int_0^\cdot H_s d \S_s
\]
for the {\em value process} associated to the trading strategy \(H \in L (\S)\).
\begin{definition}
	For \(c \in \bR_+\), a trading strategy \(H \in L (\S)\) is called {\em \(c\)-admissible} if \(\P\)-a.s. \(V^H \geq - c\). Further, we call a trading strategy {\em admissible} if it is \(c\)-admissible for some \(c \in \bR_+\).
\end{definition}

Next, we recall the definition of the classical no-arbitrage notion
\emph{NA}.

\begin{definition}[NA]\label{def:100224a1}
	We say that a strategy $H\in L(\S)$ realizes \emph{arbitrage} if
	\begin{enumerate}
		\item[\textup{(i)}]
		$H$ is admissible,
		
		\item[\textup{(ii)}]
		$V^H_\infty = \lim_{t \to \infty} V^H_t$ exists \(\P\)-a.s.,
		
		\item[\textup{(iii)}]
		$\P(V^H_\infty \ge0)=1$ and $\P(V^H_\infty > 0) > 0$.
	\end{enumerate}
	We say that the \emph{NA} condition holds for the market \((\mathbb{B}, \S)\) if there is no strategy realizing arbitrage.
\end{definition}

We remark that, in continuous-time models, it is necessary to consider admissible strategies, as non-admissible arbitrages (i.e., strategies satisfying only (ii)--(iii)) exist practically in any interesting model (e.g., in the classical Black--Scholes model).

\begin{definition}[ACLMM]
	We call a probability measure \(\Q\) on \((\Omega, \cF)\)
	an \emph{absolutely continuous local martingale measure}
	for the market \((\mathbb{B}, \S)\) if \(\Q\ll \P\)
	and \(\S\) is a \(\Q\)-local martingale.
	We use the abbreviation \emph{ACLMM} in the following.
\end{definition}

A complete stochastic characterization of NA for a finite time horizon within a Brownian setting was established by Levental and Skorohod \cite{LevSko1995}.
A necessary condition for NA for general continuous prices processes was given by Delbaen and Schachermayer \cite{DS1995}. We recall it in the following theorem.

\begin{theorem}\label{theo: FTAP NA}
	For a financial market \((\mathbb{B}, \S)\), 
	$$
	\text{NA}
	\quad\Longrightarrow\quad
	\text{there exists an ACLMM}.
	$$
\end{theorem}

As we are not aware of an explicit example in the literature showing that the converse is not true, we now present one.

\begin{example}\label{ex:290724a1}
Fix some $T\in(0,\infty)$, which will play the role of a finite time horizon, and consider a stochastic basis
\(\mathbb{B} = (\Omega, \cF, (\cF_t)_{t \geq 0}, \P)\)
that supports an $(\cF_t)_{t \geq 0}$-$\P$-Brownian motion $W=(W_t)_{t\ge0}$.
Define the process $Y = (Y_t)_{t \geq 0}$ with the dynamics
\begin{align*}
dY_t
&=
\begin{cases}
-dt+Y_t\,dW_t,&t<T_0(Y)\triangleq\inf\{t\ge0\colon Y_t=0\},
\\
dt,&t\ge T_0(Y),
\end{cases}
\\[1mm]
Y_0
&=
x_0>0.
\end{align*}
Throughout, we use the usual convention $\inf\emptyset\triangleq\infty$.
Notice that $Y$ is \emph{not} a strong Markov process.

In the following, we consider the market $(\mathbb B,Y_{\cdot\wedge T})$.
First of all, we claim that $H\triangleq\1_{(T_0(Y)\wedge T,T]}$
realizes arbitrage in this market.
Indeed, the identity
\[
V^H = \int_0^\cdot \1_{(T_0 (Y) \wedge T, T]} (s) \, ds
\] 
shows that \(H\) is admissible, that \(\P\)-a.s. \(V^H_\infty = T - T_0 (Y) \wedge T\) and that \(\P (V^H_\infty \geq 0) = 1\). Further, \(\P (V^H_\infty> 0) = \P(T_0 (Y) < T) > 0\) follows from \cite[Theorem~1.1]{bruggeman}. In summary, \(H\) realizes arbitrage and consequently, NA fails.

Next, we show that there exists an ACLMM in the market $(\mathbb B,Y_{\cdot\wedge T})$.
As the stopped process $Y_{\cdot\wedge T_0(Y)}$ is a diffusion, \cite[Example 3.15]{criensurusov23} implies that the market $(\mathbb B,Y_{\cdot\wedge T_0(Y)\wedge T})$ satisfies NA.
Thus, by Theorem~\ref{theo: FTAP NA}, there is an ACLMM $\Q$ in the market $(\mathbb B,Y_{\cdot\wedge T_0(Y)\wedge T})$.
Define the process 
$$
B_t\triangleq\int_0^t \frac1{Y_s}\,dY_s,\quad t<T_0(Y)\wedge T,
$$
and observe that $B$ is a continuous $(\cF_t)_{t \geq 0}$-$\Q$-local martingale on $[0,T_0(Y)\wedge T)$ with $\langle B,B\rangle^\Q_t=\langle B,B\rangle^\P_t=t$, $t<T_0(Y)\wedge T$.
We conclude from \cite[Exercise IV.3.28]{RY} that $B$ is an $(\cF_t)_{t \geq 0}$-$\Q$-Brownian motion on the stochastic interval $[0,T_0(Y)\wedge T)$.
By virtue of the dynamics
$$
dY_t=Y_t\,dB_t,\quad t<T_0(Y)\wedge T,
$$
this readily implies that  $\Q(T_0(Y)\ge T)=1$.
We conclude that $\Q$ is also an ACLMM for the market~$(\mathbb B,Y_{\cdot\wedge T})$.
\end{example}

For a finite time horizon, building upon \cite{DS1995}, Kabanov and Stricker \cite{KabStr2005} and Strasser \cite{Strasser2005} provided even necessary and sufficient conditions for NA.
We recall them in the next theorem.

\begin{theorem}\label{th:210324a1}
	Consider a financial market $(\mathbb B,\S)$ such that \(\S = \S_{\cdot \wedge T}\) for some time horizon $T \in (0, \infty)$.
	The following are equivalent:
	\begin{enumerate}
		\item[\textup{(i)}]
		The market satisfies NA.
		
		\item[\textup{(ii)}]
		For every stopping time $\sigma\le T$ there exists an ACLMM $\siQ$ for the market $(\siB,\siS)$ with
		$$
		\siQ\sim\P\text{ on }\cF_\sigma,
		$$
		where $\siB\triangleq(\Omega, \cF,(\sicF_t)_{t\geq 0},\P)$,
		$\sicF_t\triangleq\cF_{\sigma+t}$,
		$\siS\triangleq(\siS_t)_{t\geq 0}$,
		$\siS_t\triangleq\S_{\sigma+t}$.
	\end{enumerate}
\end{theorem}

Theorem~\ref{th:210324a1} explains that NA entails much more than only the existence of an ACLMM, namely the existence of a whole family of ACLMMs that is associated to shifted market models. Further, the result shows that this stronger condition is equivalent to NA. The question remains open whether for certain market models the sole existence of an ACLMM could also be sufficient for the NA condition. In the following section we introduce a diffusion framework where NA is often equivalent to the existence of an ACLMM. We will also identify the precise additional properties that are needed for this equivalence to hold.

\section{NA and ACLMMs in canonical general diffusion models}\label{sec: main}
First, we introduce our diffusion market model and afterwards we present our main results. 
In this section, we consider the canonical diffusion setting on the path space (cf., e.g., \cite[Definitions V.45.1,~V.45.2]{RW2}).
In Section~\ref{sec: main2} below, we show that our main results, essentially, remain valid also in the (more general) non-canonical setting.
See Discussion~\ref{disc:091024a1} for a precise statement.

\subsection{The Setup}
Let \(J \subset \bR\) be a bounded or unbounded, closed, open or half-open interval.
The interior of $J$ is denoted by $J^\circ$.
Let \(\Omega\triangleq C(\mathbb R_+; \bR)\) be the space of continuous functions \(\mathbb{R}_+ \to \bR\). The coordinate process on \(\Omega\) is denoted by \(\X\), i.e., \(\X_t (\omega) = \omega(t)\) for \(t \in \mathbb{R}_+\) and \(\omega \in \Omega\). 
We also set \(\mathcal{F} \triangleq \sigma (\X_s, s \geq 0)\) and \(\mathcal{F}_t \triangleq \bigcap_{s > t}\sigma (\X_r, r \leq s)\) for all~\(t \in \mathbb{R}_+\). 

A map \((J \ni x \mapsto \P_x)\) from \(J\) into the set of probability measures on \((\Omega, \mathcal{F})\), is called a
\emph{regular continuous strong Markov proccess} or a
\emph{general diffusion}
(with state space $J$) if the following hold:
\begin{enumerate}
	\item[\textup{(i)}] $\P_x(\X_0 = x) = \P_x(C(\mathbb R_+;J))=1$ for all $x\in J$;
	\item[\textup{(ii)}] the map \(x \mapsto \P_x(A)\) is measurable for all \(A \in \mathcal{F}\);
	\item[\textup{(iii)}] the \emph{strong Markov property} holds, i.e., for any stopping time \(\tau\) and any \(x \in J\), the kernel \(\P_{\X_\tau}\) is the regular conditional \(\P_x\)-distribution of \((\X_{t + \tau})_{t \geq 0}\) given \(\mathcal{F}_{\tau}\) on \(\{\tau < \infty\}\), i.e., \(\P_x\)-a.s. on~\(\{\tau < \infty\}\)
	\[  
	\P_x ( \X_{\cdot + \tau} \in d\omega \mid \mathcal{F}_\tau) = \P_{\X_\tau} (d \omega);
	\]
	\item[\textup{(iv)}] the diffusion is \emph{regular}, i.e., for all $x\in J^\circ$ and $y\in J$,
	\[
	\P_x(T_y<\infty)>0,
	\]
	where $T_y\triangleq\inf\{t\ge0:\X_t=y\}$ (with the usual convention $\inf\emptyset\triangleq\infty$).
\end{enumerate} 

It is well-known that a general diffusion \((x \mapsto \P_x)\) is characterized by two deterministic objects, the scale function \(\s\) and the speed measure \(\m\). The former is a strictly increasing continuous function from \(J\) into \(\bR\), while the latter is a measure on \((J, \mathcal{B} (J))\) that
satisfies \(\m ([a, b]) \in (0, \infty)\) for all \(a, b \in J^\circ\) with \(a < b\). We call the pair \((\s, \m)\) the \emph{characteristics of the general diffusion}. For precise definitions and more details on these concepts we refer to the seminal monograph \cite{itokean74} by K.~It\^o and H.~P.~McKean. More gentle introductions can be found in \cite{breiman1968probability,freedman,kallenberg,RY,RW2}. For a condensed overview we refer either to Chapter~2 of the book \cite{borodin_salminen} by Borodin and Salminen or to Section~2.2 of our previous paper \cite{criensurusov22}.

The class of general diffusions is very rich as the following examples illustrate.

\begin{example}[SDEs under the Engelbert--Schmidt conditions] \label{ex: SDE}
	Let \(\Y\) be a (possibly explosive)
	unique in law weak
	solution to the SDE
	\[
	d \Y_t = \mu (\Y_t) dt + \sigma (\Y_t) d W_t, \quad \Y_0 = x_0, 
	\]
	where \(\mu \colon J^\circ \to \bR\) and \(\sigma \colon J^\circ \to \bR\)
	are Borel functions
	satisfying the Engelbert--Schmidt conditions 
	\[
	\forall x\in J^\circ\colon
	\sigma(x)\ne0
	\qquad\text{and}\qquad
	\frac{1+|\mu|}{\sigma^2} \in L^1_\textup{loc} (J^\circ),
	\]
	and the accessible boundary points are stipulated to be absorbing.
	We refer to the paper \cite{engel_schmidt_3} or the monograph \cite[Chapter 5.5]{KaraShre} for a detailed discussion of SDEs under the Engelbert--Schmidt conditions. 
	
	It is well-known (\cite[Corollary 4.23]{engel_schmidt_3}) that SDEs under the Engelbert--Schmidt conditions give rise to a regular strong Markov family (possibly with a state space \(J \subset [- \infty, + \infty]\)) with scale function
	\begin{align*}
		\s (x) = \int^x \exp \Big\{ - \int^y \frac{2 \mu (z)}{\sigma^2 (z)} \, dz \, \Big\} \, dy, \quad x \in J^\circ,
	\end{align*} 
	and speed measure 
	\[
	\m (dx) = \frac{dx}{\s' (x) \sigma^2 (x)} \text{ on } \mathcal{B}(J^\circ).
	\]
	At this point, we stress that \(\Y\) is not necessarily a semimartingale. The simplest possible problem is that \(\Y\) may reach \(\infty\) or \(- \infty\) in finite time with positive probability.
	We refer to \cite{MU2015} for a thorough discussion of the semimartingale property of explosive solutions to~SDEs.
\end{example}

\begin{example}[Sticky Brownian motion] \label{ex: sticky BM}
	Another interesting class of diffusions are SDEs with stickiness. The most prominent example is the sticky Brownian motion\footnote{More precisely,
		Brownian motion with state space $\bR$ and sticky at zero.},
	which is a solution \(\Y\) to the (unique in law) system 
	\[
	d \Y_t = \1_{\{ \Y_t \not = 0\}} d W_t, \quad \1_{\{\Y_t = 0\}} dt = \rho\, d L^0_t (\Y),
	\]
	where \(\rho > 0\) is a so-called stickiness parameter and \(L^0 (\Y)\) is the semimartingale local time of the solution \(\Y\) in zero. For a discussion of this representation we refer to the paper \cite{EngPes}.
	
	The sticky Brownian motion is a general diffusion on natural scale, i.e., with the identity as scale function,
	with state space $\bR$
	and speed measure 
	\[
	\m (dx) = dx  + \rho \, \delta_0 (dx).	
	\]
	At this point, we notice that the sticky Brownian motion cannot be realized as a solution to an SDE like in Example~\ref{ex: SDE}, as the speed measure in Example~\ref{ex: SDE} is always absolutely continuous w.r.t. the Lebesgue measure.
	It is clear that one can also consider diffusions with more than one 
	sticky point.
\end{example}

\begin{example}[General diffusion market with skewness] \label{ex: skew}
	Another interesting class are diffusions with skewness. The most basic example is the skew Brownian motion, which is a solution \(\Y\) of the (unique in law) equation
	\[
	d \Y_t = d W_t + (2 \alpha - 1) d \ell^0_t (\Y), 
	\] 
	where \(\alpha \in (0,1) \setminus \{\frac{1}{2}\}\) is the so-called skewness parameter and \(\ell^0 (\Y)\) is the symmetric semimartingale local time of \(\Y\) in zero. 
	
	It is well-known
	(see \cite{HS}, \cite[Appendix 1.12]{borodin_salminen} or \cite[Exericse~X.2.24]{RY})
	that \(\Y\) is a general diffusion with state space \(\bR\), scale function 
	\[
	\s (x) = \begin{cases} (1 - \alpha) x, & x \geq 0, \\ \alpha x, & x < 0,\end{cases} 
	\] 
	and speed measure
	\[
	\m (dx) = \frac{dx}{v_\alpha (x)}
	\qquad\text{with}\qquad
	v_\alpha (x) = \begin{cases} 1 - \alpha, & x \geq 0, \\ \alpha, & x < 0.\end{cases}
	\]
	As the scale function is not continuously differentiable, the skew Brownian motion cannot be realized as a solution to an SDE as in Example~\ref{ex: SDE}.
\end{example}

\begin{setting}
Our inputs are a state space $J$, a scale function $\s$ and a speed measure $\m$.
They define a general diffusion $(J\ni x\mapsto\P_x)$.
The stochastic basis \(\mathbb{B}_{x_0} \triangleq (\Omega, \cF, (\cF_t)_{t \geq 0}, \P_{x_0})\) with some $x_0\in J^\circ$ serves as our underlying setup. 
\end{setting}

In this section,
we work under the following standing assumption.

\begin{SA} \label{SA: 1}
For all (equivalently, for some) \(x_0 \in J^\circ\), the coordinate process \(\X\) is a semimartingale on the stochastic basis \(\mathbb{B}_{x_0}\).
\end{SA}

\begin{remark} \label{rem: dc function SA}
It is well-known that the semimartingale property of $\X$ is merely a condition on the inverse scale function \(\s^{-1}\). In particular, it is independent of the starting value (when restricted to the interior of the state space).\footnote{Of course, when started in an absorbing boundary point, the coordinate process is constant and always a semimartingale.} We frequently use the fact that the semimartingale property of \(\X\) implies that the inverse scale function \(\s^{-1}|_{\s(J^\circ)}\) is a so-called \emph{dc function}, i.e., admits a representation as difference of two convex functions.
In case \(\s(J) = \bR\), this property is even necessary and sufficient. 
We refer to \cite[Section~5]{CinJPrSha} for more details. 
\end{remark}

\subsection{Main results}\label{sec:setting}
The following theorem reveals the precise connection between NA and the existence of an ACLMM for the financial market model introduced in the previous subsection.
We consider a finite or infinite time horizon \(T \in (0, \infty]\) and we use the notation \(l \triangleq \inf J\) and \(r \triangleq \sup J\).

\begin{theorem} \label{theo: main1}
Fix any $T\in(0,\infty]$.
The following are equivalent:
\begin{enumerate}
	\item[\textup{(a)}] NA holds for the market \((\mathbb{B}_{x_0}, \X_{\cdot \wedge T})\) for some \(x_0 \in J^\circ\).
			\item[\textup{(b)}] NA holds for the market \((\mathbb{B}_{x_0}, \X_{\cdot \wedge T})\) for all \(x_0 \in J^\circ\).
					\item[\textup{(c)}] 
							\begin{enumerate} [label=\textup{(c\arabic*)},ref=\textup{(c\arabic*)}]
							\item \label{c1} For some \(x_0 \in J^\circ\), there exists an ACLMM for the market \((\mathbb{B}_{x_0}, \X_{\cdot \wedge T})\).
							\item \label{c2} The scale function $\s$ is a dc function on $J^\circ$.
							\item \label{c3} Every finite boundary point $b\in\{l,r\}\cap\bR$ is either inaccessible or absorbing for~$\X$.
						\end{enumerate}
										\item[\textup{(d)}] 
						\begin{enumerate} [label=\textup{(d\arabic*)},ref=\textup{(d\arabic*)}]
							\item \label{d1} For all \(x_0 \in J^\circ\), there exists an ACLMM for the market \((\mathbb{B}_{x_0}, \X_{\cdot \wedge T})\).
							\item \label{d2} Every finite boundary point $b\in\{l,r\}\cap\bR$ is either inaccessible or absorbing for~$\X$.
						\end{enumerate}
\end{enumerate}
\end{theorem} 

It is instructive to emphasize the difference between conditions (c) and~(d): we do not need to state that $\s$ is a dc function when we have the existence of an ACLMM \emph{for all} starting points $x_0\in J^\circ$, whereas we need this regularity condition on $\s$ when the existence of an ACLMM is known only \emph{for some} $x_0\in J^\circ$.
Below we show that, in the case of a finite time horizon $T<\infty$, the formulation of Theorem~\ref{theo: main1} is sharp.
In contrast, for the infinite time horizon \(T = \infty\), we can sharpen Theorem~\ref{theo: main1} by removing (c2) and~(c3) from~(c) (and, consequently, by removing (d2) from~(d)), which is stated in Theorem~\ref{theo: main2} below.
Altogether, the difference between the Theorems~\ref{theo: main1} and~\ref{theo: main2} illustrates how delicate the whole picture truly is.

\begin{theorem} \label{theo: main2}
Fix any $x_0\in J^\circ$.
The following are equivalent: 
\begin{enumerate}
\item[\textup{(i)}]
NA holds for the market \((\mathbb{B}_{x_0}, \X)\).
\item[\textup{(ii)}]
There exists an ACLMM for the market \((\mathbb{B}_{x_0}, \X)\).
\end{enumerate}
Moreover, if the equivalent conditions \textup{(i)}--\textup{(ii)} hold for some $x_0\in J^\circ$, then they hold for all $x_0\in J^\circ$.
\end{theorem} 

The above theorems illustrate how close NA and the existence of an ACLMM are in our general diffusion framework. At least for the infinite time horizon, it is always equivalent to NA. Furthermore, as in many situations of interest (e.g., in the SDE setting of Example~\ref{ex: SDE}) the scale function \(\s\) is a priori a dc function and accessible boundaries are stipulated to be absorbing, we thus identify many natural subframeworks of interest where, irrespective of the time horizon, the existence of an ACLMM is equivalent to NA.

\begin{remark}
A comparison of Theorems \ref{theo: FTAP NA} and~\ref{th:210324a1} reveals that, in any financial market~$(\mathbb B,\S)$, NA implies the existence of an ACLMM $\Q$ with the additional property $\Q\sim\P$ on $\cF_0$.
In our general diffusion framework, this additional property comes automatically
due to Blumenthal's zero-one law. Indeed, if \(\Q\) is an ACLMM for the market \((\mathbb{B}_{x_0}, \X_{\cdot \wedge T})\), the absolute continuity $\Q\ll\P_{x_0}$ and Blumenthal's zero-one law even yield that \(\Q=\P_{x_0}\) on \(\cF_0\).
\end{remark}

\begin{proof}[Proof of Theorem~\ref{theo: main1}]
	We establish the following path of implications:
	\begin{align*}
		\textup{(a)} \Longrightarrow \textup{(c)} \Longrightarrow \textup{(b)} \Longrightarrow \textup{(d)} \Longrightarrow \textup{(a)}.
	\end{align*}
	First, (a)~\(\Longrightarrow\)~(c1) follows from Theorem~\ref{theo: FTAP NA} and (a)~\(\Longrightarrow\) (c2) + (c3) is due to \cite[Theorem 3.9]{criensurusov23}.
In the case $T=\infty$, we only need to observe that NA for $(\mathbb B_{x_0},X)$ implies NA for $(\mathbb B_{x_0},X_{\cdot\wedge T'})$, with any $T'\in(0,\infty)$, and to apply \cite[Theorem 3.9]{criensurusov23} to
$(\mathbb B_{x_0},X_{\cdot\wedge T'})$.\footnote{This approach seems to be somewhat indirect.
We mention that, in the case $T=\infty$, the direct way would be to apply \cite[Theorem 3.19]{criensurusov23}, which deals with the infinite-horizon case.
Contrary to \cite[Theorem 3.9]{criensurusov23}, it is assumed in \cite[Theorem 3.19]{criensurusov23} that $X$ is \emph{not} on natural scale.
Therefore, on this way, the case of natural scale needs to be treated separately.
While the latter is, of course, not a problem, this requires a bit more explanation.
That is why we try to apply the finite-horizon result \cite[Theorem 3.9]{criensurusov23} whenever it is possible.}
This establishes the first implication (a)~\(\Longrightarrow\)~(c). 

\smallskip 
Next, we assume that (c) holds.
If $X$ is on natural scale, then, due to~(c3), for any $x_0\in J^\circ$, $X$ is a $\P_{x_0}$-local martingale, hence $\P_{x_0}$ itself is an \emph{equivalent} local martingale measure for the market $(\mathbb B_{x_0},X_{\cdot\wedge T})$.
It follows that $(\mathbb B_{x_0},X_{\cdot\wedge T})$ satisfies NFLVR, let alone NA.
Thus, if $X$ is on natural scale, then (b) holds.
Now, assume that $X$ is not on natural scale.
By (c1), there exists an \(x_0 \in J^\circ\) and an ACLMM \(\Q_{x_0}\) for the market \((\mathbb{B}_{x_0}, \X_{\cdot \wedge T})\). Because \(\s\) is a dc function by (c2), the right-derivative 
	\[
	\s'_+ (x) \triangleq \lim_{h \searrow 0} \, \frac{\s (x + h) - \s (x)}{h}, \quad x \in J^\circ, 
	\]
	is well-defined and
	\(d\tilde{\m} \triangleq \s'_+ d \m\) is a valid speed measure on \(J^\circ\) by \cite[Lemma 5.10]{criensurusov23}. In particular, there exists a diffusion \((x \mapsto \tilde{\P}_{x})\) on natural scale with speed measure \(\tilde{\m}\) that is absorbed in the boundaries of \(J\) whenever they are accessible (which is a property of \(\tilde{\m}|_{J^\circ}\)). 
	Again because \(\s\) is a dc function, we can apply \cite[Lemma 5.11]{criensurusov23} and conclude that
\(\Q_{x_0} = \tilde{\P}_{x_0}\) on $\mathcal F_{T-\varepsilon}$, for any $\varepsilon\in(0,T)$.
Notice that, if $T=\infty$, this implies that $\Q_{x_0} = \tilde{\P}_{x_0}$ on the whole \(\sigma\)-field $\mathcal F$.
Take an arbitrary $T'\in(0,T)$ in the case $T<\infty$ and $T'=\infty$ in the case $T=\infty$.
Now, the preceding discussion yields $\tP_{x_0}\ll\P_{x_0}$ on $\mathcal F_{T'}$.
Together with~(c3), this allows us to conclude from Corollary~2.15 (when $T<\infty$) resp.~2.16 (when $T=\infty$) in \cite{criensurusov22} that the deterministic conditions for NA from Theorem~3.9 (when $T<\infty$) resp.~3.19 (when $T=\infty$) in \cite{criensurusov23} are satisfied.

In particular, these conditions are independent of the initial value. As a consequence, (b) holds. 
	
	\smallskip 
	We now assume that (b) holds. Part (d1) follows from Theorem~\ref{theo: FTAP NA} and (d2) is due to \cite[Theorem 3.9]{criensurusov23}. Hence, (d) holds. 
	
	\smallskip 
	Finally, we assume that (d) holds and aim to conclude (a). Thanks to \cite[Lemma 5.12]{criensurusov23}, (d1) implies that \(\s\) is a dc function on $J^\circ$. Hence, (c) holds and consequently, also~(b). As (b)~\(\Longrightarrow\)~(a) is trivial, the proof is complete.
\end{proof}

\begin{proof}[Proof of Theorem~\ref{theo: main2}]
The implication (i) $\Longrightarrow$ (ii) follows from Theorem~\ref{theo: FTAP NA}.
The last claim follows from the equivalence (a) $\Longleftrightarrow$ (b) in Theorem~\ref{theo: main1}.
It remains to prove the implication (ii) $\Longrightarrow$ (i) for a fixed $x_0\in J^\circ$.
To this end, by Theorem~\ref{theo: main1}, it suffices to prove that (ii) entails that
\begin{itemize}
\item[(a)]
$\s$ is a dc function on $J^\circ$ and
\item[(b)]
every finite boundary point $b\in\{l,r\}\cap\bR$ is either inaccessible or absorbing for~$X$.
\end{itemize}
This is the program for the remainder of this proof.

\smallskip
We first set
\(
\zeta \triangleq \inf \{t \geq 0 \colon \X_t \notin J^\circ\}.
\)
Recalling that \(\q\triangleq\s^{-1}\) is a dc function on $\s(J^\circ)$ (see Remark~\ref{rem: dc function SA}) and that \(\s (X_{\cdot \wedge \zeta})\) is a local \(\P_{x_0}\)-martingale (\cite[Corollary~V.46.15]{RW2}), we can apply the generalized It\^o formula (\cite[Theorem~29.5]{kallenberg}) to \(X_{\cdot \wedge \zeta} = \s^{-1} (\s (X_{\cdot \wedge \zeta}))\) and obtain \(\P_{x_0}\)-a.s., for all \(t < \zeta\), 
	\[
	d \X_t = \q'_- (\s (\X_t)) d \s (\X_t) + \tfrac{1}{2}\, d\int L^x_t (\s (\X)) \q'' (dx), 
	\] 
	where \(\q'_-\) denotes the left hand derivative 
	\[
	\q_-' (x) = \lim_{h \nearrow 0} \frac{\q(x + h) - \q(x)}{h}, \quad x \in J^\circ, 
	\] 
	and \(\q''\) denotes the associated second derivative measure given by 
	\[
	\q'' ( [a, b)) = \q_-' (b) - \q'_- (a), \quad a, b \in J^\circ, \, a < b.
	\] 
	Consequently, \(\P_{x_0}\)-a.s., for all \(t < \zeta\), 
	\begin{align}\label{eq: formula}
	\langle \X, \X\rangle_t = \int_0^t \big( \q'_- (\X_s) \big)^2 d \langle \s (\X), \s (\X) \rangle_s.
	\end{align} 
	
In the following we argue that \(t \mapsto \langle \X, \X\rangle_t\) is \(\P_{x_0}\)-a.s. continuous and strictly increasing on \([0, \zeta)\). Continuity is clear as \(\X\) is a continuous \(\P_{x_0}\)-semimartingale. We discuss that the map is also strictly increasing. 
Notice that any nonempty open interval \(I \subset J^\circ\) contains another nonempty interval \(I' \subset I\) such that \(\q'_- > 0\) on \(I'\). Indeed, there exists a \(c \in I\) such that \(\q'_- (c) > 0\), as otherwise \(\q\) would be constant on \(I\), which is not the case (as it is strictly increasing). Then, by the left-continuity of \(\q'_-\), there exists some \(d \in I\), $d<c$, such that \(\q'_- > 0\) on \([d, c]\).
As \(\s (\X_{\cdot\wedge\zeta})\) is a local \(\P_{x_0}\)-martingale
that has \(\P_{x_0}\)-a.s. no interval of constancy (this is inherited from the corresponding property of Brownian motion and \cite[Theorem~33.9]{kallenberg}), we conclude from \cite[Exercise~3, p.~415]{kallenberg} that \([0, \zeta)\ni t \mapsto \langle \s (\X), \s (\X)\rangle_t\) is strictly increasing and consequently, in view of \eqref{eq: formula} and by the positivity property of \(\q'_-\) as discussed above, the same is true for \([0, \zeta) \ni t \mapsto \langle \X, \X \rangle_t\).

For a while, we work on the stochastic interval
$[0,\langle X,X\rangle_\zeta)$ and
define $\P_{x_0}$-a.s. continuous and strictly increasing time change $\tau$ and time-changed process $B$ by the formulas
\[
\tau_t \triangleq \inf \{s \geq 0 \colon \langle \X, \X \rangle_s > t \}, \quad t\in[0,\langle X,X\rangle_\zeta),
\]
and
$$
\B_t\triangleq \X_{\tau_t}, \quad t\in[0,\langle X,X\rangle_\zeta).
$$
Let \(\Q_{x_0}\) be an ACLMM for the market \((\mathbb{B}_{x_0}, \X)\).
By a variant of the Doeblin, Dambis, Dubins--Schwarz theorem (see \cite[Exercise V.1.18]{RY}), the process \(\B\) is a \(\Q_{x_0}\)-Brownian motion on the stochastic interval
$[0,\langle X,X\rangle_\zeta)$.
For the rest of this proof, when we work under $\Q_{x_0}$ (but not $\P_{x_0}$), we extend $B$ to the whole time axis $[0,\infty)$ by stopping it at $\langle X,X\rangle_\zeta$.

As discussed above, \([0, \zeta) \ni t \mapsto \langle \X, \X \rangle_t\) and \([0, \langle \X, \X\rangle_\zeta) \ni t \mapsto \tau_t\) are \(\P_{x_0}\)-a.s. inverse to each other.
Hence, \(\P_{x_0}\)-a.s., for all \(a < x_0 < b\) with \(a, b \in J^\circ\),
	\[
	\overline T_a(B) \wedge \overline T_b (B) = \langle \X, \X \rangle_{\overline T_a\wedge \overline T_b},
	\]
where, for $c\in J^\circ$, we use the notation
$$
\overline T_c\triangleq\inf\{t\in[0,\zeta) \colon X_t=c\}
\quad\text{and}\quad
\overline T_c(B)\triangleq\inf\{t\in[0,\langle X,X\rangle_\zeta) \colon B_t=c\}.
$$
Considering some sequences
$\{a_n,b_n:n\in\mathbb N\}\subset J^\circ$
with
\(a_n \searrow \inf J\) and \(b_n\nearrow \sup J\), we conclude that \(\P_{x_0}\)-a.s. 
	\[
	\zeta (B) \triangleq \lim_{n\to\infty}
	\big(\overline T_{a_n}(B) \wedge \overline T_{b_n} (B)\big)
	= \langle \X, \X\rangle_\zeta.
	\]
As \(\Q_{x_0}\ll \P_{x_0}\), this equality also holds \(\Q_{x_0}\)-a.s.
Recalling our previous insight that $B$ is a $\Q_{x_0}$-Brownian motion stopped at $\langle X,X\rangle_\zeta$, we now conclude that $B$ is a $\Q_{x_0}$-Brownian motion absorbed at $\inf J$ and $\sup J$.
	
	Because \(\s (\X_{\cdot\wedge\zeta})\) is a local \(\P_{x_0}\)-martingale, Jacod's theorem (\cite[Theorem~III.3.13]{JS}) yields that \(\s (\X_{\cdot\wedge\zeta})\) is a \(\Q_{x_0}\)-semimartingale and, by \cite[Theorem~10.16]{Jacod}, the time-changed process 
	\(
	\s (\X_{\tau_\cdot\wedge\zeta}) = \s (\B)
	\)
is also a \(\Q_{x_0}\)-semimartingale.
By an obvious adjustment of \cite[Theorem~5.9]{CinJPrSha}, we conclude that \(\s|_{J^\circ}\) is a dc function.

It remains to establish that finite boundary points are either inaccessible or absorbing.
As we already proved that $\s$ is a dc function on $J^\circ$, we can apply \cite[Lemma 5.11]{criensurusov23} and conclude that $\Q_{x_0}=\tP_{x_0}$,
where $(x\mapsto\tP_x)$ is a general diffusion on natural scale with speed measure $\s'_+d\m$ on $J^\circ$ that is absorbed in the boundaries of $J$ whenever they are accessible.
Thus, $\tP_{x_0}\ll\P_{x_0}$ and \cite[Corollary 2.16]{criensurusov22} yields that \(\tP_{x_0}\) and \(\P_{x_0}\) have the same boundary behavior.\footnote{In case \(\tP_{x_0} = \P_{x_0}\) this is clear, and when \(\tP_{x_0} \not = \P_{x_0}\) one can apply \cite[Corollary 2.16]{criensurusov22}.}
This completes the proof.
\end{proof}

Finally, we construct examples showing that neither of the conditions (c2) and~(c3) can be removed from part~(c) of Theorem~\ref{theo: main1}
(and, similarly, condition~(d2) cannot be dropped).
As we already know from Theorem~\ref{theo: main2}, for such examples, we necessarily need to consider the case of a finite time horizon.

We start with a short preparation.
Define the set $\mathbb D\triangleq\{2^{-n} \colon n\in\mathbb N\}$ and the function
\(\hat{\f} \colon (0, \infty) \to (0, \infty)\) by the formula
\[
\hat{\f} (x) \triangleq
\begin{cases}
\frac{|\log_2 (x)|}{x},&x\in\mathbb D,
\\[1mm]
\frac{1}{x|\log_2 (x)|^2},&x\in(0,1/2]\setminus\mathbb D,
\\[1.5mm]
\,2,&x\in(1/2,\infty).
\end{cases} 
\]
Let \(\f \colon (0,\infty) \to (0, \infty)\) be a smoothened version of \(\hat{\f}\) such that 
\begin{enumerate}
\item[(i)]
\(\f\) is smooth on \((0, \infty)\);

\item[(ii)]
$\f=\hat\f$ on $\mathbb D\cup(1/2,\infty)$ and
\begin{equation}\label{eq:031024a1}
\frac{1}{x|\log_2 (x)|^2}
\le\f(x)\le
\frac{|\log_2 (x)|}{x},\quad
x\in(0,1/2]\setminus\mathbb D;
\end{equation}

\item[(iii)]
$\int_0^\varepsilon \f(x)\,dx<\infty$ for some (equivalently, for all) $\varepsilon\in(0,\infty)$;

\item[(iv)]
in each interval \([2^{-n-1}, 2^{-n}]\), $n\in\mathbb N$,
the function \(\f\) first decreases and then increases.
\end{enumerate} 
In the following, \(T \in (0, \infty)\) is a finite time horizon.

\begin{example}[Condition~(c2) in Theorem~\ref{theo: main1} cannot be dropped]\label{ex:031024a1}
We take the interior of the state space \(J^\circ \triangleq \bR\)
and define the scale function \(\s \colon \bR \to \bR\) by 
\[
\s (x) \triangleq	\int_0^x \f \, (|y|) \, dy, \quad x \in \bR. 
\]
As, clearly, $\s(-\infty)=-\infty$ and $\s(\infty)=\infty$, the boundary points $\pm\infty$ will be inaccessible for any general diffusion with the scale function $\s$ and any valid speed measure on $\mathcal B(\bR)$, that is, we will also have $J=\bR$.
As speed measure we take
\[
\m (dx) \triangleq dx + \sum_{n = 1}^\infty \frac{1}{n^2}
\big(\delta_{2^{-n}} (dx) + \delta_{-2^{-n}} (dx)\big)
\text{ on }\mathcal B(\bR).
\]
It is easy to see that this is indeed a valid speed measure.
We now consider the general diffusion $(\bR\ni x\mapsto\P_x)$ with the scale function $\s$ and speed measure~$\m$.

Observe that \(\s\) is not a dc function on $\bR$, because \(\s' = \f\) on \((0, \infty)\) and \(\lim_{x \searrow 0} \f (x) = \infty\). However, its inverse \(\q\triangleq\s^{-1}\) is a dc function on $\bR$. Indeed, it is easy to see that \(\q\) is absolutely continuous with absolutely continuous derivative \(1 / \f (| \q |)\). To conclude that \(\q\) is even a dc function it suffices to understand that \(1 / \f (|\q|)\) is locally of finite variation. Clearly, this is the case on \(\bR \setminus \{0\}\).
Furthermore, the lower bound in~\eqref{eq:031024a1}, item~(iv) preceding this example and the estimate
		\[
		\sum_{n = 1}^\infty 2^{-n} |\log_2 (2^{-n})|^2
		= \sum_{n = 1}^\infty 2^{-n} n^2 < \infty 
		\]  
imply that \(1 / \f (|\q|)\) is of finite variation around the origin.
In particular, this means that Standing Assumption~\ref{SA: 1} holds.

As $\s$ is not a dc function on $\bR$, the equivalent conditions (a)--(d) in Theorem~\ref{theo: main1} are violated. Clearly, (c3) holds.
In the following, we show that for any \(x_0 \in \bR \setminus \{0\}\) an ACLMM exists for the market $(\mathbb B_{x_0},X_{\cdot\wedge T})$, that is, (c1)~holds.
The conclusion will be that (c2) cannot be dropped.

We take some \(x_0 > 0\), the case \(x_0 < 0\) being symmetric, 
		and construct the ACLMM explicitly. For that, consider the measure
		\[
		\tilde{\m} (dx) \triangleq \f (x)\, \m (dx) \text{ on } \mathcal{B} (I), \ I \triangleq (0, \infty).
		\]
It is obvious that $\tilde\m([a,b])\in(0,\infty)$ for all $0<a<b<\infty$.
Furthermore, for all \(\varepsilon \in (0, \infty)\), we have 
		\begin{align*} 
			\int_{ (0, \varepsilon) } x\, \tilde{\m} (dx) \geq \sum_{n = 1}^\infty 2^{- n} \, \f (2^{-n}) \, \frac{1}{n^2} \, \1_{\{2^{-n} \, <\, \varepsilon \}} = \sum_{n = 1}^\infty \, \frac{1}{n} \, \1_{\{2^{-n} \, <\, \varepsilon\}} = \infty. 
		\end{align*}
This means that the origin is inaccessible for the general diffusion on natural scale with speed measure \(\tilde{\m}\) and the interior of the state space $I$,
cf. \cite[Proposition 16.43]{breiman1968probability}, that is, its state space will also be~$I$.
Let \((I \ni x \mapsto \tilde{\P}_x)\) be such a diffusion.
We also introduce another general diffusion $(\bR_+\ni x\mapsto\P'_x)$ by the formula
$\P'_x\triangleq\P_x\circ X_{\cdot\wedge T_0}^{-1}$, where
$T_0\triangleq\inf\{t\ge0 \colon X_t=0\}$.
As the state spaces of the general diffusions $(x \mapsto \tP_x)$ and $(x \mapsto \P'_x)$ have the same interior,
the main results from \cite{criensurusov22} apply.
In particular, \cite[Corollary 2.15]{criensurusov22} yields that
\(\tilde{\P}_{x_0} \ll \P'_{x_0}\) on \(\cF_T\).
Recalling that the origin is inaccessible under $\tP_{x_0}$ and, hence,
$\tP_{x_0}(T_0>T)=1$,
we can conclude that
\(\tilde{\P}_{x_0} \ll \P_{x_0}\) on \(\cF_T\).
Further, as the diffusion $(I\ni x\mapsto\tP_x)$ is on natural scale,
under \(\tilde{\P}_{x_0}\), the process \(X\) is a local martingale.
We thus obtain that the restriction
$\tP_{x_0}|\mathcal F_T$ of $\tP_{x_0}$ to the $\sigma$-field $\mathcal F_T$ is an ACLMM for the market
$$
\big((\Omega,\mathcal F,(\mathcal F_t)_{t\in[0,T]},\P_{x_0}|\mathcal F_T),(X_t)_{t\in[0,T]}\big).
$$
It remains to lift this ACLMM to our global setting $(\mathbb B_{x_0},X_{\cdot\wedge T})$.
To this end, let $Z_T\triangleq d(\tP_{x_0}|\mathcal F_T)/d(\P_{x_0}|\mathcal F_T)$ be the Radon-Nikodym derivative of the restrictions of the measures to $\mathcal F_T$.
Then the measure $\bar\P_{x_0}\triangleq Z_T\cdot\P_{x_0}$ on the whole $\mathcal F$ is an ACLMM for the market~\((\mathbb{B}_{x_0}, X_{\cdot \wedge T})\).
\end{example}

\begin{remark}
In the previous example, we also have the following.

\begin{enumerate}
\item[(i)]
\(\tilde{\P}_{x_0}\not \ll \P_{x_0}\) on \(\cF\) (cf. \cite[Corollary 2.16]{criensurusov22}).
This is as suggested by Theorem~\ref{theo: main2}.

\smallskip
\item[(ii)]
No ACLMM exists for the market $(\mathbb B_{x_0},X_{\cdot\wedge T})$ with \(x_0 = 0\).
This follows from Theorem~\ref{theo: main1} because otherwise (d) in Theorem~\ref{theo: main1} would be satisfied.
\end{enumerate}
\end{remark}

\begin{example}[Neither of conditions (c3) and~(d2) in Theorem~\ref{theo: main1} can be removed]\label{ex:031024a2}
Let $(\bR\ni x\mapsto\P_x)$ be the general diffusion constructed in Example~\ref{ex:031024a1}.
We define the general diffusion $(\bR_+\ni x\mapsto\Q_x)$ by the formula
\begin{equation}\label{eq:041024a1}
\Q_x\triangleq\P_x\circ|X|^{-1},\quad x\in\bR_+.
\end{equation}
In other words, this is a diffusion with state space $\bR_+$, scale function
		\[
		\s (x) \triangleq	\int_0^x \f \, (y) dy, \quad x \in \bR_+,
		\]
and speed measure 
		\[
		\m (dx) \triangleq dx + \sum_{n = 1}^\infty \frac{1}{n^2}\, \delta_{2^{-n}} (dx)\text{ on } \mathcal{B} (\bR_+)
		\]
(e.g., see \cite[Lemma B.11]{criensurusov22}).
Clearly, for this diffusion the boundary point $0$ is reflecting (instantaneously, as $\m(\{0\})=0$).
In the following we show that the diffusion $(\bR_+\ni x\mapsto\Q_x)$ provides the necessary example.

In Example~\ref{ex:031024a1} we proved that $X$ is a $\P_x$-semimartingale for all $x\in\bR$.
Hence, $|X|$ is a $\P_x$-semimartingale for all $x\in\bR$.
Consequently, $X$ is a $\Q_x$-semimartingale for all $x\in\bR_+$, which means that Standing Assumption~\ref{SA: 1} holds.

As the origin is reflecting for the diffusion $(x\mapsto\Q_x)$,
the equivalent conditions (a)--(d) in Theorem~\ref{theo: main1} are violated.
Clearly, $\s$ is a dc function on $(0,\infty)$, that is, (c2) holds.
It remains to show that, for any initial value $x_0>0$, there exists an ACLMM for the market
$(\mathbb B_{x_0},X_{\cdot\wedge T})$, where, with a slight abuse of notation\footnote{In
$\mathbb B_{x_0}$ we now use $\Q_{x_0}$ instead of $\P_{x_0}$, as the notation ``$\P_{x_0}$'' is already occupied (see~\eqref{eq:041024a1}).},
$\mathbb B_{x_0}\triangleq(\Omega,\mathcal F,(\mathcal F_t)_{t\ge0},\Q_{x_0})$.
The conclusion will be that (c3) and~(d2) cannot be dropped.

Fix any $x_0>0$. The construction of an ACLMM is now performed in the same way as in Example~\ref{ex:031024a1}.
More precisely, with $\tP_{x_0}$ defined exactly as in Example~\ref{ex:031024a1}, we get that $\tP_{x_0}\ll\Q_{x_0}$ on $\cF_T$. Then the measure
$$
\bar\P_{x_0}\triangleq
\frac{d(\tP_{x_0}|\mathcal F_T)}{d(\Q_{x_0}|\mathcal F_T)}\cdot\Q_{x_0}
\text{ on }\mathcal F
$$
is an ACLMM for the market $(\mathbb B_{x_0},X_{\cdot\wedge T})$.
\end{example}

\section{NA and ACLMMs in non-canonical general diffusion models}\label{sec: main2}

The purpose of this section is to present versions of our main results for general diffusion markets beyond the canonical setting. 

\begin{setting}
Our inputs are a state space $J$, a scale function $\s$ and a speed measure $\m$ as described in Section~\ref{sec: main}.
Further, we consider a filtered probability space
$\mathbb B\triangleq(\Omega,\mathcal F,(\mathcal F_t)_{t\ge0},\P)$
with a right-continuous filtration that supports a regular continuous strong Markov process
$Y=(Y_t)_{t\ge0}$
with state space $J$, scale function $\s$, speed measure $\m$ and a deterministic starting point $x_0\in J^\circ$.
In the above context, the strong Markov property refers to the filtration $(\mathcal F_t)_{t\ge0}$.
\end{setting}

In this section, we work under the following standing assumption, which is merely a condition on the inverse scale function (recall Remark~\ref{rem: dc function SA}).

\begin{SA} \label{SA: 2}
$Y$ is a semimartingale on $\mathbb B$.
\end{SA}

As before, we use the notation $l\triangleq\inf J$ and $r\triangleq\sup J$.
In this (non-canonical) setting, the main results are as follows.

\begin{theorem}\label{th:091024a1}
Consider a finite time horizon $T\in(0,\infty)$.
Then, the following are equivalent:
\begin{enumerate}
\item[\textup{(a)}]
NA holds for the market \((\mathbb{B}, Y_{\cdot\wedge T})\).

\item[\textup{(b)}]
\begin{enumerate}[label=\textup{(b\arabic*)},ref=\textup{(b\arabic*)}]
\item
There exists an ACLMM for the market \((\mathbb{B}, Y_{\cdot\wedge T})\).

\item
The scale function $\s$ is a dc function on $J^\circ$.

\item
Every finite boundary point $b\in\{l,r\}\cap\bR$ is either inaccessible or absorbing for~$Y$.
\end{enumerate}
\end{enumerate}
\end{theorem}

Recalling Examples \ref{ex:031024a1} and~\ref{ex:031024a2}, we observe that neither of the conditions (b2) and~(b3) can be removed from part~(b) of Theorem~\ref{th:091024a1}.
But again, we do not need these conditions in the case of the infinite time horizon:

\begin{theorem}\label{th:091024a2}
Then the following are equivalent:
\begin{enumerate}
\item[\textup{(i)}]
NA holds for the market \((\mathbb{B}, Y)\).

\item[\textup{(ii)}]
There exists an ACLMM for the market \((\mathbb{B}, Y)\).
\end{enumerate}
\end{theorem}

For what follows we introduce the notation
$(\mathcal F^Y_t)_{t\ge0}$
for the right-continuous filtration of $Y$, i.e.,
$\mathcal F^Y_t\triangleq\bigcap_{s>t}\sigma(Y_r,r\le s)$, $t\in\bR_+$,
and set
$\mathbb B^Y\triangleq(\Omega,\mathcal F,(\mathcal F^Y)_{t\ge0},\P)$.

\begin{discussion}\label{disc:091024a1}
In comparison with the canonical setting, here we have no family of measures corresponding to different starting points.
Therefore, the surprising effect that parts (c) and~(d) in Theorem~\ref{theo: main1} are both optimally formulated and equivalent cannot be discussed within the non-canonical setting.

On the other hand, for non-canonical spaces, the results in Theorems \ref{theo: main1} and~\ref{theo: main2} translate only to Theorems \ref{th:091024a1} and~\ref{th:091024a2} with $\mathbb B$ replaced by $\mathbb B^Y$
(more precisely, this follows from results on change of the space from \cite[Section 10.2a]{Jacod}, in particular,
\cite[Theorem 10.37]{Jacod} and \cite[Proposition 10.38~(b)]{Jacod}).
Thus, the essential point in Theorems \ref{th:091024a1} and~\ref{th:091024a2} is that the filtration $(\mathcal F_t)_{t\ge0}$ is allowed to be an enlargement of~$(\mathcal F^Y_t)_{t\ge0}$.
We only need that $Y$ is strongly Markovian w.r.t. $(\mathcal F_t)_{t\ge0}$.
\end{discussion}

\begin{proof}[Proof of Theorem~\ref{th:091024a1}]
Assume~(a).
Then Theorem~\ref{theo: FTAP NA} implies~(b1), while \cite[Theorem 3.9]{criensurusov23} yields (b2) and~(b3).

Now, assume~(b).
Recall that a \emph{continuous} local martingale (say, $M=(M_t)_{t\ge0}$)
w.r.t. some filtration (say, $(\mathcal G_t)_{t\ge0}$)
remains a local martingale w.r.t any filtration $(\mathcal H_t)_{t\ge0}$
that lies between the natural filtration of $M$ and $(\mathcal G_t)_{t\ge0}$.
Hence, by~(b1), there exists an ACLMM for the market $(\mathbb B^Y,Y_{\cdot\wedge T})$.
By \cite[Theorem 10.37]{Jacod}, there exists an ACLMM for the canonical market $(\mathbb B_{x_0},X_{\cdot\wedge T})$, where we use the notation from Section~\ref{sec: main}.
Proceeding in the same way as in the proof of the implication (c)~$\Longrightarrow$~(b) in Theorem~\ref{theo: main1}, we obtain the deterministic condition for NA from Theorem~3.9 in \cite{criensurusov23} (that theorem is for a general non-canonical setting) and conclude that (a) in Theorem~\ref{th:091024a1} is satisfied.
This completes the proof.
\end{proof}

\begin{proof}[Proof of Theorem~\ref{th:091024a2}]
The implication (i)~$\Longrightarrow$~(ii) follows from Theorem~\ref{theo: FTAP NA}.
For the converse, assume that (ii) holds.
In the same way as in the previous proof we get that there exists an ACLMM for the canonical market $(\mathbb B_{x_0},X)$.
Applying first the implication (ii)~$\Longrightarrow$~(i) in Theorem~\ref{theo: main2} and then the implication (a)~$\Longrightarrow$~(c) in Theorem~\ref{theo: main1} for $T=\infty$, we infer that conditions (c2) and~(c3) of Theorem~\ref{theo: main1} are satisfied.
Proceeding in the same way as in the proof of the implication (c)~$\Longrightarrow$~(b) in Theorem~\ref{theo: main1} for $T=\infty$, we obtain the deterministic condition for NA from Theorem~3.19 in \cite{criensurusov23}
(to be precise, the case of natural scale needs to be treated separately,
cf. the proof of the implication (c)~$\Longrightarrow$~(b) in Theorem~\ref{theo: main1}).
Thus, (i) in Theorem~\ref{th:091024a2} is satisfied, and the proof is complete.
\end{proof}

\bibliographystyle{plain}

\end{document}